\documentclass[aps,pra,twocolumn,amsfonts,amssymb,amsmath,showpacs,
floatfix,nofootinbib,citesort]{revtex4-2}
\usepackage{mathrsfs}
\usepackage{amsfonts}
\usepackage{amstext}
\usepackage{amsmath}
\usepackage{amssymb}
\usepackage{bm}
\usepackage{CJK}
\usepackage{bbm}
\usepackage[dvips]{graphicx}
\def\qed{\leavevmode\unskip\penalty9999 \hbox{}\nobreak\hfill
	\quad\hbox{\leavevmode  \hbox to.77778em{%
			\hfil\vrule   \vbox to.675em%
			{\hrule width.6em\vfil\hrule}\vrule\hfil}}
	\par\vskip3pt}

\usepackage{amssymb}
\usepackage{graphicx}
\usepackage{graphics}
\usepackage{amsmath}
\usepackage{amsthm}
\usepackage{color}
\usepackage{dsfont}
\usepackage{textcomp}
\definecolor{darkred}  {rgb}{0.5,0,0}
\definecolor{darkblue} {rgb}{0,0,0.5}
\definecolor{darkgreen}{rgb}{0,0.5,0}
\usepackage{hyperref}
\hypersetup{
	pdftitle = {QRT Proposal},
	pdfauthor = {},
	colorlinks = true,
	urlcolor  = blue,         
	linkcolor = red,     
	citecolor = blue,    
	filecolor = darkred       
}
\usepackage{mathtools}
\def\ra{\rangle}
\def\la{\langle}
\def\bb{\mathbb}

\newtheorem{theorem}{Theorem}

\newtheorem{lemma}{Lemma}

\newcommand{\bea}{\begin{eqnarray}}
	\newcommand{\eea}{\end{eqnarray}}
\newcommand{\be}{\begin{equation}}
	\newcommand{\ee}{\end{equation}}
\newcommand{\ba}{\begin{equation}\begin{aligned}}
		\newcommand{\ea}{\end{aligned}\end{equation}}

\newcommand{\beax}{\begin{eqnarray*}}
	\newcommand{\eeax}{\end{eqnarray*}}
\newcommand{\bex}{\begin{equation*}}
	\newcommand{\eex}{\end{equation*}}

\newtheorem{definition}{Definition}
\theoremstyle{remark}

\newtheorem{example}{Example}

\def\be{\begin{equation}}
	\def\ee{\end{equation}}

\newcommand{\mE}{\mathcal{E}}

\newcommand{\mR}{\mathcal{R}}

\newcommand{\mI}{\mathcal{I}}
\newcommand{\mH}{\mathcal{H}}

\newcommand{\mM}{\mathcal{M}}

\newcommand{\mD}{\mathcal{D}}

\newcommand{\mS}{\mathcal{S}}

\newcommand{\tr}{{\rm Tr}}

\newcommand{\mbR}{\mathbb{R}}

\newcommand{\rIm}{{\rm Im}}
\newcommand{\rRe}{{\rm Re}}

\begin{document}
	

\preprint{APS/123-QED}
\begin{CJK*}{GB}{gbsn}
\title{Imaginarity witness\\}
		

\author{Zhiqi Liang}
\author{Yufan Lin}		
\author{Yu Guo}
\email{guoyu3@aliyun.com}
\author{Yuqi Sun}

\affiliation{School of Mathematical Sciences, Inner Mongolia University, Hohhot, Inner Mongolia 010021, People's Republic of China}
\affiliation{Inner Mongolia Key Laboratory of Mathematical Modeling and Scientific Computing, Inner Mongolia University, Hohhot, Inner Mongolia 010021, People's Republic of China}

		
\begin{abstract}
Imaginarity has shown to be an important resource in quantum information. The witness theory of quantum resource, such as entanglement witness, coherence witness, and imaginarity witness, has been established, in particular entanglement witness and coherence witness have been extensively explored. We present here another class of imaginarity witnesses beyond the one in [{Phys. Lett. A} \textbf{530}, 130135 (2025)]. Within our framework, any nonreal Hermitian operator under the reference basis is an imaginarity witness and only finite of them can detect all the imaginarity states. As the common approaches that were explored in the witness theories of entanglement and coherence, we then explore the relations between these imaginarity witnesses in different cases: (i) when they can detect common imaginary states, (ii) when they can detect the same sets of imaginary states, and (ii) when they obey the finer relation. Finally, we define an imaginarity measure in terms of witnesses, termed witnessed imaginarity, and prove that it coincides with both the trace norm of imaginarity and the robustness of imaginarity.

\end{abstract}

\maketitle
\end{CJK*}


\section{Introduction}


In classical physics, complex numbers are often used to simplify models of oscillatory motion and wave dynamics~\cite{Smith W F.2010,Smith W F.2013,Tatzko S2019}. It seems that they play a much deeper role in quantum physics since it has been demonstrated to offer advantages in many quantum information tasks, such as state discrimination~\cite{WuKD2021prl,Herzog2002pra}, hiding and masking~\cite{ZhuH2021prr}, quantum machine learning~\cite{Sajjan2023prr}, pseudorandomness~\cite{Koh2023arxiv}, multiparameter quantum metrology~\cite{Carollo2018srep,Carollo2019jsm,Miyazaki2022quantum}, linear-optical implementations~\cite{Jones2023pra,Menssen2017prl,Shchesnovich2018pra},
Kirkwood--Dirac quasiprobabilities~\cite{Wagner2024qst,Budiyono2023pra1,Budiyono2023pra2,Budiyono2023jmp},
and weak-value amplification and related phenomena~\cite{Wagner2023pra,Kedem2012pra,Dixon2009prl,Hosten2008science,Brunner2010prl,Hofmann2011pra,Kunjwal2019pra}.
These results indicate that imaginarity is indispensable in quantum theory.

Ever since Hickey and Gour~\cite{HickeyGour2018} introduced the resource theory of imaginarity in 2018, in which quantum states with complex matrix elements in a fixed basis are considered resources, the resource theory of imaginarity has received extensive attention. To characterize the imaginarity of a quantum state, one typically requires full knowledge of its density matrix elements. But, in experiments, this information is usually obtained via quantum state tomography, which is a bureaucratic procedure. This raised the studies on the imaginarity witness, which was first introduced in~\cite{Li N 2025 pla}. A Hermitian operator $W$ is called an imaginarity witness if (i) $\tr(W\delta)\geq0$ for any real state $\delta$, and (ii) there exists an imaginary state $\rho$ such that $\tr(W\rho)<0$. Another question that has drawn much attention in the resource theory of imaginarity is how to quantify imaginarity. Various imaginarity measures have been proposed so far, including the contractive metric induced imaginarity measure~\cite{HickeyGour2018}, the robustness of imaginarity~\cite{HickeyGour2018}, the geometric measure of imaginarity~\cite{WuKD2021pra}, the relative entropy of imaginarity~\cite{Xue2021qip}, the weight of imaginarity~\cite{Xue2021qip}, and the Tsallis relative entropy of imaginarity~\cite{Xu J 2024pla}, etc.

Indeed, witnesses have been developed extensively in convex quantum resource theory, such as entanglement witnesses~\cite{Wu Y C2006pla,Hou J2010pra,Lewenstein2000pra,Wu Y Cl2007pra,Hansen2015 ijqi,Brandao2005pra,Qi2012} and coherence witnesses~\cite{Napoli2016prl,Wang B H2021Entropy,Li M S2024pra,Ren H2017alp,Zhu X N2024jpa,Wang2021qip,Z. Ma2021}. In the context of coherence witnesses, Li \emph{et al.} constructed in~\cite{Li M S2024pra} a series of coherence witnesses in the light of the traces of observables, which was shown to be stronger than the previous coherence witnesses~\cite{Napoli2016prl}. Ref.~\cite{Zhu X N2024jpa} discussed coherence witnesses based on the maximum and minimum diagonal elements of given observables and showed that these witnesses can detect all the coherence in a more general manner. This motivates us to improve the construction of imaginarity witnesses. Consequently, we propose a new class of imaginarity witnesses which can detect the imaginary by the minimal and maximal eigenvalues of the real part of the witness. Moreover, we show that these witnesses indeed enhance the ability to detect imaginarity and finite these imaginarity witnesses can detect all imaginarity in the systems.

Another issue along this line is analyzing the relationship between different witnesses, which have been explored extensively for the case of entanglement and coherence~\cite{Wu Y C2006pla,Hou J2010pra,Lewenstein2000pra,Wu Y Cl2007pra,Hansen2015 ijqi,Brandao2005pra,Qi2012,Napoli2016prl,Wang B H2021Entropy,Li M S2024pra,Ren H2017alp,Zhu X N2024jpa,Wang2021qip,Z. Ma2021}. In Refs.~\cite{Wu Y C2006pla,Hou J2010pra,Lewenstein2000pra,Wu Y Cl2007pra,Wang B H2021Entropy,Li M S2024pra,Wang2021qip}, the condition for different witnesses can detect some common states, and the condition for two witnesses can detect the same states, etc, have been investigated. We thus discuss such an issue for imaginarity witness we proposed. On the other hand, the resource witnesses can induce corresponding measure. For instance, the witnessed coherence coincides with the $l_1$-norm of coherence.~\cite{Ren H2017alp}; the $l_1$-norm of imaginarity can be represented by imaginarity witnesses~\cite{Li N 2025 pla}. As one may expect, we give an imaginarity measure by means of the imaginarity witnesses and prove that it coincides with both the trace norm of imaginarity and the robustness of imaginarity.

The rest of this paper is organized as follows. In Sec.~\ref{sec2}, we review the resource theory of imaginarity and two imaginarity measures that are closely related with our work. In Sec.~\ref{sec3}, we introduce a new type of imaginarity witnesses and prove that only finite such witnesses can detect all the imaginary states. We analyze in Sec.~\ref{sec5} the condition under which different witnesses can detect some common imaginary states. In Sec.~\ref{sec6}, we derive the condition for two witnesses that can detect exactly the same set of imaginary states. Sec.~\ref{sec7} investigates the finer relation of witnesses. Sec.~\ref{sec8} discusses the witnessed imaginarity. In particular, we find that it coincides with both the trace norm of imaginarity and the robustness of imaginarity. Sec.~\ref{sec9} contains some conclusions.


\section{Preliminaries}\label{sec2}
	

Throughout this paper, we denote by $\mS$ the set of all quantum states acting on the system described by the Hilbert space $\mH$ with $\dim\mH=d\geqslant2$. Recall that, the free states in resource theory of imaginarity are {\it real states}, i.e., states with a real density matrix $\langle i |\rho |j\rangle\in{\mathbb R}$ under a given reference basis $\{|i\ra\}_{m=0}^{d-1}$ of $\mathcal H$. 
Let $\mR$ denote the set of real states with respect to the reference basis. All the states in $\mS\setminus\mR$ are referred to as imaginary states. We denote by
\bea
\rRe(\rho):= \frac{\rho+\rho^T}{2}\ \text{and}\ \rIm(\rho):= \frac{\rho-\rho^T}{2\rm i}
\eea
the real part and the imaginary part of the quantum state $\rho$ under the reference basis, respectively. Here $\rm i:=\sqrt{-1}$ is the imaginary unit and $\rho^T:=\sum_ {i,j=1}^{d}\la{j}|\rho|i\ra|i\ra\la{j}|$ denotes the transpose of $\rho$ with respect to the reference basis.

We denote by $\mathscr{H}$ the set of all Hermitian operator acting on $\mH$. An imaginarity witness~\cite{Li N 2025 pla} is a Hermitian operator $W\in\mathscr{H}$ satisfying (i) $\tr(W\delta)\geq0$ for all real states $\delta\in\mR$, and (ii) there exists an imaginary state $\rho$ such that $\tr(W\rho)<0$. $W$ is an imaginarity witness if and only if $\rRe(W)\geq0$ under the reference basis and it has some negative eigenvalue. Hereafter, we assume that all the operations are taken under this reference basis. Another imaginarity witnesses, called stringent imaginarity witnesses, strengthens the condition (i) to $\tr(W\delta)=0$ for all $\delta\in\mR$ and relaxes the condition (ii) to $\tr(W\rho)\neq0$ for some state $\rho$. Such a class of witnesses can also detect all imaginary states~\cite{Li N 2025 pla}. A imaginarity witness $W$ is a stringent imaginarity witness if and only if $\rRe(W)=0$.

The corresponding free operation is called {\it real operation}, which can be represented as $\mE(\rho) =\sum_j K_j\rho K_j^{\dag}$ with the
Kraus operators $K_j$'s satisfy $\langle m|K_j|n\rangle\in{\mathbb R}$. A nonnegative function ${\mathcal I}$ on $\mS(d)$ is
called an imaginarity measure if it admits the following conditions (I1) to (I4)~\cite{WuKD2021prl,HickeyGour2018,WuKD2021pra}. 

(I1) Non-negativity: $\mathcal I (\rho) \geqslant 0$, and $\mathcal I (\rho) = 0$ for any real state $\rho$. 

(I2) Monotonicity: $\mathcal I (\mE(\rho))\leqslant  \mathcal I (\rho)$ whenever $\mE$ is a real operation. 

(I3) Probabilistic monotonicity: $\sum_j p_j{\mathcal I}(\rho_j)\leqslant {\mathcal I}(\rho),$ where $p_j=\tr (K_j\rho K_j^{\dag})$, $\rho_j=\frac {1}{p_j}K_j\rho K_j^{\dag}$, for all $\{K_j\}$ with $\sum_j K_j^{\dag}K_j=I$ and $K_j$ are real operators. 

(I4) Convexity: ${\mathcal I}(\sum_jp_j\rho_j)\leqslant \sum_jp_j{\mathcal I}(\rho_j)$ for any ensemble $\{p_j,\rho_j\}$. \\
Note that (I3) and (I4) together imply (I2). In Ref.~\cite{Xue2021qip}, another condition was discussed, i.e., 

(I5) additivity for direct sum states: ${\mathcal I}(p\rho_1 \oplus (1-p)\rho_2) = p{\mathcal I}(\rho_1)+(1-p){\mathcal I}(\rho_2)$,~$0<p<1$. (I3) and (I4) are equivalent to (I2) and (I5)~\cite{Xue2021qip}.

We introduce two imaginarity measures that are particularly pertinent to our work. For any state $\rho$, the trace norm of imaginarity~\cite{HickeyGour2018} was defined as  
\bea
\mI_{tr}(\rho)=\min_{\delta\in\mR}\|\rho-\delta\|_{\tr}=\frac{1}{2}\|\rho-\rho^T\|_{\tr}, 
\eea
where $\|A\|_{\tr}=\tr\sqrt{A^\dagger A}$. The robustness of imaginarity~\cite{HickeyGour2018,WuKD2021pra} was defined by
\bea
\mI_{R}(\rho)=\min_\tau\left\lbrace s\geq0\ \big |\frac{\rho+s\tau}{1+s}\in\mR\right\rbrace =\|\rIm(\rho)\|_{\tr}.
\eea

\section{Imaginarity witnesses based on the maximal-minimal eigenvalues of its real part} \label{sec3}

Note that the imaginarity witnesses in Ref.~\cite{Li N 2025 pla} are based on the non-negative real part of witness $W$. 
Analogous to the coherence witness based on the maximal-minimal diagonal elements of the observable~\cite{Zhu X N2024jpa}, here we introduce a new type of imaginarity witnesses based on the maximum and minimum eigenvalues of $\rRe(W)$. We give the following Lemma first.

\begin{lemma}
	For any $W\in\mathscr{H}$ with $\lambda_{\max}(\rRe(W))=M$ and $\lambda_{\min}(\rRe(W))=m$, we have $\tr(W\delta)\in[m,M]$ for any $\delta\in\mR$.
\end{lemma}

\begin{proof}
	For any real state $\delta$, we have $\tr(W\delta)=\tr({\rRe}(W)\delta)$. Further,
	one has 
	\beax
	\lambda_{\min}(\rRe(W))\leq\tr(W\delta)=\tr(\rRe(W)\delta) \leq\lambda_{\max}(\rRe(W)),
	\eeax 
	hence proved.                                    
\end{proof}

That is, if $W\in\mathscr{H}$ with $\lambda_{\max}(\rRe(W))=M$ and $\lambda_{\min}(\rRe(W))=m$, then $\tr(W\rho)<m$ or $\tr(W\rho)>M$ implies that $\rho$ is an imaginary state. In fact, we have the following result.

\begin{theorem}\label{th1}
	If $m\leqslant M$, then, for any imaginary state $\rho$, there exists $W\in\mathscr{H}$ with $\rIm(W)\neq0$, $\lambda_{\min}(\rRe(W))=m$ and $\lambda_{\max}(\rRe(W))=M$, such that $\tr(W\rho)\notin [m, M]$.
\end{theorem}

\begin{proof}
	For any imaginary state $\rho$, we assume $\rIm(\rho_{kl})=\rIm(\la k|\rho|l\ra)\neq0$ for some $k\not=l$. We let
	\bea 
	W_{k,l}:=\frac{\rm i}{2}(|k\ra\la l|-|l\ra\la k|).
	\eea 
	If $m=M$, we take $W=W_{k,l}+mI$, where $I$ is the identity operator on $\mH$. Clearly, $\tr(W\rho)=\rIm(\rho_{kl})+m\neq m$. If $m\neq M$, we take $W_0\in\mathscr{H}$ with $\lambda_{\min}(\rRe(W_0))=m$ and $\lambda_{\max}(\rRe(W_0))=M$. Let $W:=\frac{M+1-\tr[W_0\rho]}{\rIm(\rho_{kl})}W_{k,l}+W_0$, we have $\lambda_{\min}(\rRe(W))=m$, $\lambda_{\max}(\rRe(W))=M$ and $\tr(W\rho)=M+1>M$. This completes the proof.  	
\end{proof}

Theorem~\ref{th1} implies that, for any $m\leqslant M$, the Hermitian operators with with $\rIm(W)\neq0$, $\lambda_{\min}(\rRe(W))=m$ and $\lambda_{\max}(\rRe(W))=M$ are really imaginarity witnesses. We thus give the following new type of imaginarity witnesses.

\begin{definition}\label{D1}
	Let $W\in\mathscr{H}$ be a Hermitian operator with $\lambda_{\max}(\rRe(W))=M$ and $\lambda_{\min}(\rRe(W))=m$. We call $W$ an imaginarity witness if (i) for each real state $\delta$, $\tr(W\delta)\in[m,M]$, and (ii) there exists an imaginary state $\rho$, such that $\tr(W\rho)\notin[m,M]$.
\end{definition}

By definition, this family of imaginarity witnesses includes the ones in Ref.~\cite{Li N 2025 pla} as special case: If $m=0$, it is the general imaginary witness in~\cite{Li N 2025 pla}; if $m=M=0$, it is the stringent imaginarity witness in~\cite{Li N 2025 pla}. In Definition~\ref{D1}, $m$ and $M$ can be arbitrary real numbers with the only requirement is $m\leqslant M$.

We fix some notations for convenience. We denote by $\bb{W}_{m,M}$ the set of all imaginarity witnesses $W$'s with  $\lambda_{\max}(\rRe(W))=M$ and $\lambda_{\min}(\rRe(W))=m$, i.e.,
\be
\begin{aligned}
	\bb{W}_{m,M}&:=\{W\in\mathscr{H} \big|\lambda_{\max}(\rRe(W))=M, \ \lambda_{\min}(\rRe(W))\\
	&~~~~~=m, \exists\, \rho\in\mS,\  \text{s.t.}\ \tr(W\rho)\notin[m,M]\},
\end{aligned}
\ee
and for any $W\in \bb{W}_{m,M}$, we let
\beax
	\bb{I}_W&:=&\{\rho\in\mS|\tr(W\rho)\notin[m, M]\},\\
	\bb{I}_{W}^{m} &:=&\{\rho\in\mS\ |\ \tr(W\rho)<m\}, \\
	\bb{I}_{W}^{M} &:=&\{\rho\in\mS\ |\ \tr(W\rho)>M\}.
\eeax
Clearly, 
$\bb{I}_{W}=\bb{I}_{W}^{m}\bigcup\bb{I}_{W}^{M}$.

From the proof of Theorem~\ref{th1}, we can obtain that all imaginarity states can be detected by a set of finite imaginarity witnesses $\{W_i\}_{i=1}^n\subset\bb{W}_{m,M}$ when $M=m=a$ for any $a\in\mathbb{R}$. But for the case of $M\neq m$, any finite imaginarity witnesses $\{W_i\}_{i=1}^n\subset\bb{W}_{m,M}$ can not detect all the imaginarity states. We let $\rho_{\epsilon}=(1-\epsilon)\delta_0+\epsilon\rho'$, where $\delta_0=\frac{\rm I}{d}\in\mR$, $\rho'\in\mS\setminus\mR$, $0<\epsilon<1$. Then 
\beax
\tr(W_i\rho_{\epsilon})=(1-\epsilon)\tr(W_i\delta_0)+\epsilon\tr(W_i\rho')
\eeax
for any $i\in\{1,...,n\}$. Let $\tr(W_i\delta_0):=a_i$ and $K:=\max_i{\tr(W_i\rho')}$, and let
\bea
0<\epsilon<\min\left\lbrace \frac{a_i-m}{|K|+|a_i|+1},\frac{M-a_i}{|K|+|a_i|+1}\right\rbrace. 
\eea
Since
$-|\tr(W_i\rho')|-|a_i|\leq\tr(W_i\rho')-a_i\leq|\tr(W_i\rho')|+|a_i|$, we have $\epsilon<\frac{a_i-m}{|K|+|a_i|+1}<\frac{a_i-m}{|K|+|a_i|}$. Thus,
\beax
\tr(W_i\rho_{\epsilon})-m&=&\epsilon(\tr(W_i\rho')-a_i)+a_i-m\\&>&-(a_i-m)+a_i-m=0.
\eeax
Similarly, one can show that $M-\tr(W_i\rho_{\epsilon})>0$. However, $\rho_{\epsilon}$ is an imaginary state which can not be detected by any witness in $\{W_i\}_{i=1}^n$.

Going further, for the case of $M=m$, we discuss the minimal number of witnesses that are needed to detect all imaginary states in a qudit systems.

For the qubit system, if $M=m=a$, we let 
\bea
W&=\begin{pmatrix}
	a&a_{12}\rm i\\
	-a_{12}\rm i&a\\
\end{pmatrix},\  \ a, a_{12}\in\bb{R}, \ a_{12}\neq0.
\eea
Any qubit state can be written as $\rho=\frac{1}{2}(I+x\sigma_x+y\sigma_y+z\sigma_z)$, $x,y,z\in\bb{R}$. If $\rho$ is an imaginary state, $y\neq 0$. Then $\tr(W\rho)=a-a_{12}y\neq a$, so $\rho\in\bb{I}_W$.

For $3$-dimensional quantum system, we can construct witnesses as follows: 
\beax
W_1=\begin{pmatrix}
	a&a_{12}\rm i&0\\
	-a_{12}\rm i&a&0\\
	0&0&a\\
\end{pmatrix},~
W_2=\begin{pmatrix}
	a&0&a_{13}\rm i\\
	0&a&0\\
	-a_{13}\rm i&0&a\\
\end{pmatrix},
\eeax
\beax
W_3=\begin{pmatrix}
	a&0&0\\
	0&a&a_{23}\rm i\\
	0&-a_{23}\rm i&a\\
\end{pmatrix},
W_4=\begin{pmatrix}
	a&a_{12}\rm i&a_{13}\rm i\\
	-a_{12}\rm i&a&0\\
	-a_{13}\rm i&0&a\\
\end{pmatrix},
\eeax
\beax
W_5=\begin{pmatrix}
	a&a_{12}\rm i&0\\
	-a_{12}\rm i&a&a_{23}\rm i\\
	0&-a_{23}\rm i&a\\
\end{pmatrix},
W_6=\begin{pmatrix}
	a&0&a_{13}\rm i\\
	0&a&a_{23}\rm i\\
	-a_{13}\rm i&-a_{23}\rm i&a\\
\end{pmatrix},
\eeax
\beax		
W_7=\begin{pmatrix}
	a&a_{12}\rm i&a_{13}\rm i\\
	-a_{12}\rm i&a&a_{23}\rm i\\
	-a_{13}\rm i&-a_{23}\rm i&a\\
\end{pmatrix},	
\eeax  
where $a$, $a_{12}$, $a_{13}$, $a_{23}$ are any real number and $a_{12}a_{13}a_{23}\neq 0$. Any qutrit state can be expressed as $\rho=\frac{\rm I}{3}+\frac{1}{2}\sum_{k=1}^{8}b_kG_k$, where $G_k$s are then Gell-Mann matrices~\cite{Bertlmann2008jpa}, i.e,
\beax
G_1=\begin{pmatrix}
	0&1&0\\
	1&0&0\\
	0&0&0\\
\end{pmatrix},~
G_2=\begin{pmatrix}
	0&-\rm i&0\\
	\rm i&0&0\\
	0&0&0\\
\end{pmatrix},	
G_3=\begin{pmatrix}
	0&0&1\\
	0&0&0\\
	1&0&0\\
\end{pmatrix},
\eeax
\beax
G_4=\begin{pmatrix}
	0&0&-\rm i\\
	0&0&0\\
	\rm i&0&0\\
\end{pmatrix},	
G_5=\begin{pmatrix}
	0&0&0\\
	0&0&1\\
	0&1&0\\
\end{pmatrix},
G_6=\begin{pmatrix}
	0&0&0\\
	0&0&-\rm i\\
	0&\rm i&0\\
\end{pmatrix},
\eeax
\beax	
G_7=\begin{pmatrix}
	1&0&0\\
	0&-1&0\\
	0&0&0\\
\end{pmatrix},~
G_8=\frac{1}{\sqrt{3}}\begin{pmatrix}
	1&0&0\\
	0&1&0\\
	0&0&-2\\
\end{pmatrix},	
\eeax
and $b_k=\tr(\rho G_k)$, $k=1$, $2$, $\dots$, $8$. If a qutrit state $\rho$ is imaginary, at least one of $b_2$, $b_4$ and $b_6$ is nonzero. If $b_2\neq0$ and $b_4=b_6=0$, then $\tr(W_1\rho)=a-b_2a_{1,2}\neq a$, hence $W_1$ can detect $\rho$. $W_4$, $W_5$, and $W_7$ can also detect $\rho$. Similarly, one can easily check that:  if $b_4\neq0$ and $b_2=b_6=0$, then $W_2$, $W_4$, $W_6$, and $W_7$ are witnesses of $\rho$; if $b_6\neq0$ and $b_2=b_4=0$, then $W_3$, $W_5$, $W_6$, and $W_7$ are witnesses of $\rho$; if $b_2b_4\neq0$ and $b_6=0$, then $W_1$, $W_2$, $W_5$, and $W_6$ are its witnesses; if $b_2b_6\neq0$ and $b_4=0$, then $W_1$, $W_3$, $W_4$, and $W_6$ are its witnesses; if $b_4b_6\neq0$ and $b_2=0$, then $W_2$, $W_3$, $W_4$, and $W_5$ are its witness; if $b_2b_4b_6\neq0$, then $W_1$, $W_2$, and $W_3$ are its witnesses. Hence, all imaginary qutrit states can be detected using only three witnesses. In addition, there are $25$ possible triples of witnesses from the set $\{W_1$, $W_2$,$...$, $W_7\}$, for instance, 
$\{W_1,W_2,W_3\}$, $\{W_1,W_2,W_5\}$, $\{W_1,W_2,W_6\}$, $\{W_1$, $W_2$, $W_7\}$, $\{W_1,W_3,W_4\}$, $\{W_1,W_3,W_6\}$, $\{W_1,W_3,W_7\}$, $\{W_1,W_4,W_5\}$, $\{W_1$, $W_4$, $W_6\}$, $\{W_1,W_4,W_7\}$, $\{W_1,W_5,W_6\}$, $\{W_1,W_5,W_7\}$, $\{W_2,W_3,W_4\}$, $\{W_2$, $W_3$, $W_5\}$, $\{W_2,W_3,W_7\}$, $\{W_2,W_4,W_5\}$, $\{W_2,W_4,W_6\}$, $\{W_2,W_4,W_7\}$, $\{W_2$, $W_5$, $W_6\}$, $\{W_2,W_6,W_7\}$, $\{W_3,W_4,W_5\}$, $\{W_3,W_4,W_6\}$, $\{W_3,W_5,W_6\}$, $\{W_3$, $W_5$, $W_7\}$ or $\{W_3,W_6,W_7\}$,  that suffice to detect the all imaginary states.

By analogy, for $d$-dimensional quantum systems, we can construct witnesses of the form
\be
\begin{aligned}  
	&W_1=a I+a_{12}W_{1,2},\\
	& \qquad\dots\\
	&W_{\frac{d(d-1)}{2}}=a I+a_{d-1d}W_{d-1,d},\\
	&W_{\frac{d(d-1)+2}{2}}=a I+a_{12}W_{1,2}+a_{13}W_{1,3},\\
	& \qquad\dots\\
	&W_{d(d-1)}=a I+a_{d-2d}W_{d-2d}+a_{d-1d}W_{d-1,d},\\
	&W_{d(d-1)+1}=a I+a_{12}W_{1,2}+a_{13}W_{1,3}++a_{14}W_{1,4},\\
	& \qquad\dots\\
	&W_{2^{\frac{d(d-1)}{2}}-2}=a I+a_{13}W_{1,3}+...+a_{d-1d}W_{d-1,d},\\
	&W_{2^{\frac{d(d-1)}{2}}-1}=a I+\sum_{k,l}a_{kl}W_{k,l},
\end{aligned} 
\ee
where $a$, $a_{kl}$ are any real number and $a_{kl}\neq 0$, $k=1,...,d-1$, $k<l\leq d$. We can also check that all the imaginary states can be detected by $\frac{d(d-1)}{2}$ witnesses of these forms. For example, $W_1,...,W_{\frac{d(d-1)}{2}}$ all can detect all imaginary states. In the special case of $a=0$, theses imaginarity witnesses reduce to stringent imaginarity witnesses, which also need $\frac{d(d-1)}{2}$ to detect all imaginary states.


\section{When different witnesses can detect some common imaginary states} \label{sec5}


In this section, we discuss, for any given $m\leqslant M$, when different witnesses in $\bb{W}_{m,M}$ can detect some common imaginary states. We give the following lemma at first.

\begin{lemma}\label{L2}
	If $W_1,W_2\in\bb{W}_{m,M}$ and $W=\lambda W_1+(1-\lambda)W_2\in\bb{W}_{m,M}$, $\lambda\in (0,1)$, then $\bb{I}_W\subseteq\bb{I}_{W_1}\cup\bb{I}_{W_2}$.
\end{lemma}

\begin{proof}
	If $\rho\notin\bb{I}_{W_1}\cup\bb{I}_{W_2}$, this implies $\rho\notin\bb{I}_{W_1}$ and $\rho\notin\bb{I}_{W_2}$. By definition of $\bb{I}_W$, $\tr(W_1\rho)\in[m,M]$ and $\tr(W_2\rho)\in[m,M]$, and thus 	
	\beax
	\tr(W\rho)=\lambda\tr(W_1\rho)+(1-\lambda)\tr(W_2\rho)\in[m,M].
	\eeax 
	Hence, $\bb{I}_W\subseteq\bb{I}_{W_1}\cup\bb{I}_{W_2}$.     
\end{proof}

Note that, $W_1,W_2\in\bb{W}_{m,M}$ can not guarantee $W=\lambda W_1+(1-\lambda)W_2\in\bb{W}_{m,M}$ in general, $\lambda\in (0,1)$. For example, $W_1=\sigma_y+\frac{1}{2}\sigma_z\in\bb{W}_{-1/2,1/2}$ and $W_2=-\sigma_y+\frac{1}{2}\sigma_z\in\bb{W}_{-1/2,1/2}$, but  $W=\frac{1}{2}W_1+\frac{1}{2}W_2=\frac{1}{2}\sigma_z\notin\bb{W}_{-1/2,1/2}$ since $W$ is a real matrix.

Unlike that of the coherence witnesses~\cite{Wang B H2021Entropy} and entanglement witnesses~\cite{Wu Y C2006pla,Hou J2010pra}, the quantum states detected by the convex combination of two imaginarity witnesses may not include the common states detected by $W_1$ and $W_2$. We illustrate this fact by the following example.

\begin{example}
	Let 
	\beax
	W_1&=\begin{pmatrix}
		1&\rm i&\rm i\\
		-\rm i&2&0\\
		-\rm i&0&3\\
	\end{pmatrix},~
	W_2&=\begin{pmatrix}
		1&0&-\rm i\\
		0&2&\rm i\\
		\rm i&-\rm i&3\\
	\end{pmatrix},
	\eeax
	then $W=\lambda W_1+(1-\lambda)W_2\in\bb{W}^I_{1,3}$ for all $\lambda\in [0,1]$. We take $\lambda=\frac{1}{2}$ and 
	\beax
	\rho&=\begin{pmatrix}
		{1}/{3}&{\rm i}/{3}&{\rm i}/{3}\\
		-{\rm i}/{3}&{1}/{3}&-{\rm i}/{3}\\
		-{\rm i}/{3}&{\rm i}/{3}&{1}/{3}\\
	\end{pmatrix}.
	\eeax
	It is straightforward that $\tr(W_1\rho)=\frac{10}{3}\notin[1,3]$ and $\tr(W_2\rho)=\frac{2}{3}\notin[1,3]$, so $\rho\in\bb{I}_{W_1}\bigcap\bb{I}_{W_2}$. However, $\tr(W\rho)=2\in[1,3]$. Namely, $\bb{I}_{W_1}\bigcap\bb{I}_{W_2}\nsubseteq\bb{I}_W$. 
\end{example}

\begin{lemma}\label{L3}
	Let $W,W_1,W_2\in\bb{W}_{m,M}$ with $m\neq M$. If $\bb{I}_{W_1}\bigcap\bb{I}_{W_2}=\emptyset$ and $\bb{I}_W\subseteq\bb{I}_{W_1}\cup\bb{I}_{W_2}$, then either $\bb{I}_{W}^{m}\subseteq\bb{I}_{W_1}$ or $\bb{I}_{W}^{m}\subseteq\bb{I}_{W_2}$ and either $\bb{I}_{W}^{M}\subseteq\bb{I}_{W_1}$ or $\bb{I}_{W}^{M}\subseteq\bb{I}_{W_2}$.
\end{lemma}

\begin{proof}
	Suppose both $\bb{I}_{W_1} \bigcap\bb{I}_{W}^{m}$ and $\bb{I}_{W_2}\bigcap \bb{I}_{W}^{m}$ are nonempty. We take $\rho_i=\bb{I}_{W_i }\bigcap\bb{I}_{W}^{m}$ with $i=1,2$. Considering segment
	\bea
	[\rho_1,\rho_2]:=\{\rho_t=(1-t)\rho_1+t\rho_2 |\ 0\leq t\leq1\}.
	\eea 
	Since $\bb{I}_{W}^{m}$ is a convex set, we have 	
	\beax [\rho_1,\rho_2]\subseteq\bb{I}_{W}^{m}\subseteq\bb{I}_{W_1}\cup\bb{I}_{W_2}=\left(\bb{I}_{W_1}^{m}\bigcup\bb{I}_{W_1}^{M}\right)\bigcup\left( \bb{I}_{W_2}^{m}\bigcup\bb{I}_{W_2}^{M}\right).
	\eeax
	For simplicity, we define 
	\beax 
	\varOmega_1:=[\rho_1,\rho_2]\bigcap\bb{I}_{W_1}^{m},~
	\varOmega_2:=[\rho_1,\rho_2]\bigcap\bb{I}_{W_1}^{M},\\
	\varOmega_3:=[\rho_1,\rho_2]\bigcap\bb{I}_{W_2}^{m},~
	\varOmega_4:=[\rho_1,\rho_2]\bigcap\bb{I}_{W_2}^{M}.
	\eeax
	Thus we get
	\be 
	[\rho_1,\rho_2]\\=\varOmega_1\cup\varOmega_2\cup\varOmega_3\cup\varOmega_4,
	\ee 
	that is, $[\rho_1,\rho_2]$ is divided into four convex parts. Since $\bb{I}_{W_1} \bigcap\bb{I}_{W}^{m}$ and $\bb{I}_{W_2}\bigcap \bb{I}_{W}^{m}$ are nonempty by assumption, at least one $\Omega_i$ in both $\{\Omega_1, \Omega_2\}$ and  $\{\Omega_3, \Omega_4\}$ is nonempty. We assume with no loss of generality that $\Omega_2$ and $\Omega_3$ are nonempty. So there exists $t_0\in(0,1)$ such that $\{\rho_t |\ 0\leq t<t_0\}\subseteq\bb{I}_{W_1}^{M}$,
	$\{\rho_t|\ t_0<t\leq1\}\subseteq\bb{I}_{W_2}^{m}$ and either $\rho_{t_0}\in\bb{I}_{W_1}^{M}$ or $\rho_{t_0}\in\bb{I}_{W_2}^{m}$. 
	If $\rho_{t_0}\in\bb{I}_{W_1}^{M}$, then $\tr(W_1\rho_{t_0})>M$. This leads to, for sufficiently small $\epsilon>0$ with $t_0+\epsilon\leq1$, we have
	\beax
	\tr(W_1\rho_{t_0+\epsilon})\in [m,M] 
	\eeax 
	and 
	\beax
	\tr(W_1\rho_{t_0+\epsilon})&=&\tr(W_1\rho_{t_0})+\epsilon\left[\tr(W_1\rho_2)-\tr(W_1\rho_1)\right]\\
	&>&M
	\eeax 
	simultaneously, a contradiction. Similarly, $\rho_{t_0}\in\bb{I}_{W_2}^{m}$ leads to a contradiction as well. If three or all of $\{\Omega_1, \Omega_2, \Omega_3, \Omega_4\}$ are not empty sets, by similar arguments, it also leads to a contradiction as above. Therefore, either $\bb{I}_{W}^{m}\subseteq\bb{I}_{W_1}$ or $\bb{I}_{W}^{m}\subseteq\bb{I}_{W_2}$. A similar argument applies to $\bb{I}_{W}^{M}$. This completes the proof. 
\end{proof}

\begin{theorem}\label{th2}
	Let $\{W_i\}_{i=1}^n$ be a finite set in $\bb{W}_{m,M}$. Then the following statements hold.
	\begin{enumerate}
		\item [\rm(1)]For $m\neq M$ and any $t_i\geq0$ with $\sum_it_i=1$, if  $W=\sum_it_iW_i\in\bb{W}_{m,M}$, then $\bigcap_{i=1}^n\bb{I}_{W_i}\not=\emptyset$.    
		\item [\rm(2)] If $m=M=a$, then $\bigcap_{i=1}^n\bb{I}_{W_i}\neq\emptyset$ and it is an infinite set.
	\end{enumerate}	
\end{theorem}

\begin{proof}{\rm(1)} We consider the case of $n=2$ at first. We suppose $\bb{I}_{W_1}\bigcap\bb{I}_{W_2}=\emptyset$, $W_\lambda=\lambda W_1+(1-\lambda)W_2\in\bb{W}_{m,M}$, $\lambda\in[0,1]$. By Lemma~\ref{L2} and~\ref{L3},
	\beax
	\bb{I}_{W_\lambda}^{m}\subseteq\bb{I}_{W_1}\ \text{or}~\ \bb{I}_{W_\lambda}^{m}\subseteq\bb{I}_{W_2};~ \bb{I}_{W_\lambda}^{M}\subseteq\bb{I}_{W_1}\ \text{or}~\ \bb{I}_{W_\lambda}^{M}\subseteq\bb{I}_{W_2}.
	\eeax 
	When $\lambda$ varies from $0$ to $1$ continuously, $\bb{I}_{W_\lambda}^{m}$ also varies from $\bb{I}_{W_2}^{m}$ to $\bb{I}_{W_1}^{m}$ continuously. We take $\lambda_0:= \sup\{\lambda|\ \bb{I}_{W_\lambda}^{m}\subseteq \bb{I}_{W_2}\}$.
	
	If $\bb{I}_{W_{\lambda_0}}^{m}\subseteq\bb{I}_{W_2}$, then there exist $0<\epsilon<1-\lambda_0$ such that ${W_{\lambda_0+\epsilon}}$ is not an imaginarity witness. Otherwise, if for all $0<\epsilon<1-\lambda_0$, $\bb{I}_{W_{\lambda_0+\epsilon}}^{m}\neq\emptyset$ or $\bb{I}_{W_{\lambda_0+\epsilon}}^{M}\neq\emptyset$, then we have
	$\bb{I}_{W_{\lambda_0}}^{m}\subseteq\bb{I}_{W_2}$ and $\bb{I}_{W_{\lambda_0+\epsilon}}^{m}\subseteq\bb{I}_{W_1}$, For all $\rho\in \bb{I}_{W_{\lambda_0}}^{m}$, we have 
	\beax
	\tr[W_{\lambda_0}\rho]<m \ , \
	\tr[W_{\lambda_0}\rho]+\epsilon[\tr(W_1\rho)-\tr(W_2\rho)]\geq m,
	\eeax
	a contradiction.
	
	If $\bb{I}_{W_{\lambda_0}}^{m}\subseteq\bb{I}_{W_1}$,  
	then there exist $0<\epsilon<\lambda_0$ such that ${W_{\lambda_0-\epsilon}}$ is not an imaginarity witness. Or else,
	if for all $0<\epsilon<\lambda_0$, $\bb{I}_{W_{\lambda_0-\epsilon}}^{m}\neq\emptyset$ or $\bb{I}_{W_{\lambda_0-\epsilon}}^{M}\neq\emptyset$, then we have
	$\bb{I}_{W_{\lambda_0}}^{m}\subseteq\bb{I}_{W_1}$, $\bb{I}_{W_{\lambda_0-\epsilon}}^{m}\subseteq\bb{I}_{W_2}$,
	and for all $\rho\in \bb{I}_{W_{\lambda_0}}^{m}$, we have
	\beax
	\tr[W_{\lambda_0}\rho]<m \ , \
	\tr[W_{\lambda_0}\rho]+\epsilon[\tr(W_2\rho)-\tr(W_1\rho)]\geq m.
	\eeax
	This also leads a contradiction. Similarly, $\bb{I}_{W_\lambda}^{M}\subseteq\bb{I}_{W_1}\ \text{or}\ \bb{I}_{W_\lambda}^{M}\subseteq\bb{I}_{W_2}$ leads to a contradiction as well.

We assume that the conclusion holds for $k\leq n-1$. By induction, we need to show that the theorem holds for $n$. We show it for $n=3$ and the general case can be followed similarly.

By the case of $n=2$, 
\beax
\bb{I}_{W_1}\bigcap\bb{I}_{W_2}\neq\emptyset~\text{and}~\bb{I}_{W_1}\bigcap\bb{I}_{W_3}\neq\emptyset.
\eeax 
If 
\beax
\bb{I}_{W_1}\bigcap\bb{I}_{W_2}\bigcap\bb{I}_{W_3}=\emptyset,
\eeax   
then
\be\label{3}
\left(\bb{I}_{W_1}\bigcap\bb{I}_{W_2}\right)\bigcap\left(\bb{I}_{W_1}\bigcap\bb{I}_{W_3}\right)=\emptyset.
\ee
Let $W_\mu:=\mu W_2 +(1-\mu)W_3, \mu\in[0,1]$,
then, by Lemma~\ref{L2},
\beax
\left(\bb{I}_{W_1}^{m}\cap\bb{I}_{W_\mu}^{m}\right)\cup\left(\bb{I}_{W_1}^{m}\cap\bb{I}_{W_\mu}^{M}\right)\subseteq\left(\bb{I}_{W_1}^{m}\cap\bb{I}_{W_2}\right)\cup\left(\bb{I}_{W_1}^{m}\cap\bb{I}_{W_3}\right),\\
\left(\bb{I}_{W_1}^{M}\cap\bb{I}_{W_\mu}^{m}\right)\cup\left(\bb{I}_{W_1}^{M}\cap\bb{I}_{W_\mu}^{M}\right)\subseteq\left(\bb{I}_{W_1}^{M}\cap\bb{I}_{W_2}\right)\cup\left(\bb{I}_{W_1}^{M}\cap\bb{I}_{W_3}\right).
\eeax 
According to Eq.~\eqref{3}, we know $\bb{I}_{W_1}^{m}\bigcap\bb{I}_{W_2}$, $\bb{I}_{W_1}^{m}\bigcap\bb{I}_{W_3}$, $\bb{I}_{W_1}^{M}\bigcap\bb{I}_{W_2}$, and $\bb{I}_{W_1}^{M}\bigcap\bb{I}_{W_3}$ are disjoint. Since $\bb{I}_{W_1}^{m}\cap\bb{I}_{W_\mu}^{m}$, $\bb{I}_{W_1}^{m}\cap\bb{I}_{W_\mu}^{M}$, $\bb{I}_{W_1}^{M}\cap\bb{I}_{W_\mu}^{m}$, and $\bb{I}_{W_1}^{M}\cap\bb{I}_{W_\mu}^{M}$ are convex, we know that $\bb{I}_{W_1}^{m}\cap\bb{I}_{W_\mu}^{m}$ varies from $\bb{I}_{W_1}^{m}\bigcap\bb{I}_{W_3}^{m}$ to $\bb{I}_{W_1}^{m}\bigcap\bb{I}_{W_2}^{m}$ whenever $\mu$ varies from $0$ to $1$ continuously. We take $\mu_0:= \sup\{\mu~ |~\bb{I}_{W_1}^{m}\bigcap\bb{I}_{W_\mu}^{m}\subseteq\bb{I}_{W_1}^{m}\cap\bb{I}_{W_3}\}$. 
Using the same argument as for the case $n=2$ above, we
can conclude that ${W_{\mu_0}}$ is not an imaginarity witness, which also leads to,
\begin{align*}
	W &=\theta W_1+(1-\theta)W_{\mu_0}\\
	&=\theta W_1+(1-\theta)\mu_0 W_2+(1-\theta)(1-\mu_0)W_3
\end{align*}
is not an imaginarity witness for some $\theta\in[0,1]$, a contradiction. Similarly, when we consider $\bb{I}_{W_1}^{m}\cap\bb{I}_{W_\mu}^{M}$, $\bb{I}_{W_1}^{M}\cap\bb{I}_{W_\mu}^{m}$, and $\bb{I}_{W_1}^{M}\cap\bb{I}_{W_\mu}^{M}$, they also lead to contradictions as above. 

{\rm(2)} As both $\bb{I}_{W_1}$ and $\bb{I}_{W_2}$ are nonempty, we assume $\rho\in\bb{I}_{W_1}$, $\sigma\in\bb{I}_{W_2}$ and $\rho,\sigma\notin\bb{I}_{W_1}\bigcap\bb{I}_{W_2}$. That is, $\rho\notin\bb{I}_{W_2}$ and $\sigma\notin\bb{I}_{W_1}$, namely, we have $\tr(W_2\rho)=\tr(W_1\sigma)=a$. Then for any $\lambda\in(0,1)$, we consider $\varrho:=\lambda\rho+(1-\lambda)\sigma$. Clearly, 
\beax
\tr(W_1\varrho)&=\lambda\tr(W_1\rho)+(1-\lambda)a\neq a, \\  
\tr(W_2\varrho)&=\lambda a+(1-\lambda)\tr(W_2\sigma)\neq a.
\eeax 
So, for each $\lambda\in(0,1)$, $\varrho$ belongs to $\bb{I}_{W_1}\bigcap\bb{I}_{W_2}$. Now we suppose that the statement holds when the cardinality of the set of imaginarity witnesses is equal to $m-1$. Let $\widetilde{W}=\{W_1,W_2,...,W_m\}$, $\widetilde{W}_1=\{W_1,...,W_{m-1}\}$, $\widetilde{W}_2=\{W_2,...,W_m\}$. As the two sets $$\bb{I}_{\widetilde{W}_1}=\bigcap_{i=1}^{m-1}\bb{I}_{W_i},\ \bb{I}_{\widetilde{W}_2}=\bigcap_{i=2}^{m}\bb{I}_{W_i}$$ are nonempty, we assume $\rho\in\bb{I}_{\widetilde{W}_1}$ and $\sigma\in\bb{I}_{\widetilde{W}_2}$. If $\tr(W_m\rho)\neq a$, we have $\rho\in\bb{I}_{\widetilde{W}}=\bigcap_{i=1}^{m}\bb{I}_{W_i}$, i.e., $\bb{I}_{\widetilde{W}}\neq\emptyset$. Therefore, we assume $\tr(W_m\rho)=a$ and $\tr(W_1\sigma)=a$. In the same way, for each $\lambda\in(0,1)$, we define $\eta_\lambda:=\lambda\rho+(1-\lambda)\sigma$, then
\beax
\tr(W_1\eta_\lambda)&=&\lambda\tr(W_1\rho)+(1-\lambda)a\neq a,\\ \tr(W_m\eta_\lambda)&=&\lambda a+(1-\lambda)\tr(W_m\sigma)\neq a.
\eeax
In addition, for each $2\leqslant i \leqslant m-1$, we have $a_i:=\tr(W_i\rho)\neq a$ and $b_i :=\tr(W_i\sigma)\neq a$, i.e, $a_i:=\tr[(W_i-aI)\rho]\neq0$ and $b_i :=\tr[(W_i-aI)\sigma]\neq0$. Define 	
\bea
f_i(\lambda)&:=&\tr[(W_i-aI)\eta_\lambda]=(a_i-b_i)\lambda+b_i.
\eea
If $a_i=b_i$, then we have $\tr(W_i\eta_\lambda)=b_i\neq a$. Then $\eta_\lambda\in\bb{I}_{\widetilde{W}}$. If $a_i\neq b_i$, for each $\lambda\in(0,1)\setminus\{\frac{b_i}{a_i-b_i}\}_{i=2}^{m-1}$,  $\tr[(W_i-aI)\eta_\lambda]=f(\lambda_i)\neq0$,
$2\leqslant i\leqslant m-1$, which also holds for $i=1,m$. Therefore, all these $\eta_\lambda$'s belong to the set $\bb{I}_{\widetilde{W}}$. 

Note that, for any imaginary state $\rho$, there exists a witness $W_{\bar{\rho}}\in \bb{W}_{m,M}$ that can not detect the imaginarity of $\rho$, i.e., $\rho\notin\bb{I}_{W_{\bar{\rho}}}$. Suppose $\bigcap_{i=1}^n\bb{I}_{W_i}$ is finite, then $\bigcap_{i=1}^n\bb{I}_{W_i}=\{\rho_j\}_{j=1}^m$ for some $m\in\mathbb{N}$. This leads to    
\beax
\left(\bigcap_{i=1}^n\bb{I}_{W_i}\right)\bigcap\left(\bigcap_{j=1}^m\bb{I}_{W_{\overline{\rho_j}}}\right)=\emptyset,    
\eeax 
a contradiction. This completes the proof. 
\end{proof}

\begin{example}
Let 
\beax
W_1=\begin{pmatrix}
	1&\rm i&0\\
	-\rm i&0&0\\
	0&0&1\\
\end{pmatrix},~
W_2=\begin{pmatrix}
	1&0&\rm i\\
	0&0&0\\
	-\rm i&0&1\\
\end{pmatrix},
\eeax
\beax
W_3=\begin{pmatrix}
	1&0&0\\
	0&0&\rm i\\
	0&-\rm i&1\\
\end{pmatrix}.
\eeax
Then $W_i\in \bb{W}_{0,1}$, and for any $t_i\geq0$ with $\sum_i^3t_i=1$, $W=t_1W_1+t_2W_2+t_3W_3\in\bb{W}_{0,1}$. Taking
\beax
\rho=\begin{pmatrix}
	{1}/{3}&{\rm i}/{3}&{\rm i}/{3}\\
	-{\rm i}/{3}&{1}/{3}&{\rm i}/{3}\\
	-{\rm i}/{3}&-{\rm i}/{3}&{1}/{3}\\
\end{pmatrix},
\eeax
we get $\tr(W_1\rho)=\tr(W_2\rho)=\tr(W_3\rho)=\frac{4}{3}\notin[0,1]$, i.e, $\rho\in\bb{I}_{W_1}\bigcap\bb{I}_{W_2}\bigcap\bb{I}_{W_3}$.
\end{example}

\begin{example}
Let 
\beax
W_1=\begin{pmatrix}
	1&\rm i&\rm i\\
	-\rm i&1&0\\
	-\rm i&0&1\\
\end{pmatrix},~
W_2=\begin{pmatrix}
	1&0&-\rm i\\
	0&1&\rm i\\
	\rm i&-\rm i&1\\
\end{pmatrix},
\eeax
\beax
\rho=\begin{pmatrix}
	{1}/{3}&{\rm i}/{3}&{\rm i}/{3}\\
	-{\rm i}/{3}&{1}/{2}&0\\
	-{\rm i}/{3}&0&{1}/{6}\\
\end{pmatrix}.
\eeax
Then $W_{1,2}\in\bb{W}_{1,1}$. It is straightforward that $\tr(W_1\rho)=\frac{7}{3}\neq1$ and $\tr(W_2\rho)=\frac{1}{3}\neq1$, i.e., $\rho\in\bb{I}_{W_1}\bigcap\bb{I}_{W_2}$.
\end{example}

It is worth mentioning that the converse of Theorem~\ref{th2} $(1)$ is not ture, which is different from that of coherence witnesses~\cite{Wang B H2021Entropy} and entanglement witnesses~\cite{Wu Y C2006pla,Hou J2010pra}. Namely, although two witnesses $W_1$ and $W_2$ can detect common imaginary states, their convex combination may not be an imaginarity witness. We present the following illustrative example.

\begin{example}
We consider 
\beax
W_1=\begin{pmatrix}
	1&\rm i&0\\
	-\rm i&0&0\\
	0&0&1\\
\end{pmatrix},~
W_2=\begin{pmatrix}
	0&-\rm i&0\\
	\rm i&1&0\\
	0&0&0\\
\end{pmatrix}
\eeax
in $\bb{W}_{0,1}$ and 
\beax
\rho=\begin{pmatrix}
	{1}/{3}&{\rm i}/{3}&0\\
	-{\rm i}/{3}&{1}/{2}&0\\
	0&0&{1}/{6}\\
\end{pmatrix}.
\eeax
We have $\tr(W_1\rho)=\frac{7}{6}\notin[0,1]$ and $\tr(W_2\rho)=-\frac{1}{6}\notin[0,1]$, but $W=\lambda_0 W_1+(1-\lambda_0)W_2\notin\bb{W}_{0,1}$ whenever $\lambda_0=\frac{1}{2}$. 
\end{example}


\section{When two witnesses can detect the same imaginary states} \label{sec6}


The conditions under which two witnesses can simultaneously detect identical resource states for both entanglement and coherence have been explored in Ref.~\cite{Wu Y C2006pla,Hou J2010pra,Lewenstein2000pra,Li M S2024pra}. With the same spirit in mind, we consider in this section the counterpart for imaginarity. The main result is the following.

\begin{theorem}\label{th3}
Let $W_1$ and $W_2$ be two imaginarity witness in $\bb{W}_{m,M}$. Then the following statements are true.
\begin{enumerate}
	\item[(1)] If $m\neq M$, then $\bb{I}_{W_1}=\bb{I}_{W_2}$ if and only if $W_1=W_2$ or $W_1+W_2=(m+M)I$.
	\item[(2)] If $m=M=a$, then $\bb{I}_{W_1}=\bb{I}_{W_2}$ if and only if $W_1-aI=r(W_2-aI)$. Moreover, if there no exists some $r\neq0$, such that $W_1-aI=r(W_2-aI)$, then $\bb{I}_{W_1}\nsubseteq\bb{I}_{W_2}$ and $\bb{I}_{W_2}\nsubseteq\bb{I}_{W_1}$.
\end{enumerate}
\end{theorem}

\begin{proof}
{\rm(1)} Since 
\bex
\bb{I}_{W_1}=\bb{I}_{W_1}^{m}\bigcup\bb{I}_{W_1}^{M},~
\bb{I}_{W_2}=\bb{I}_{W_2}^{m}\bigcup\bb{I}_{W_2}^{M},
\eex
and $\bb{I}_{W_1}^{m}$ and $\bb{I}_{W_1}^{M}$ are two disjoint convex sets, so are $\bb{I}_{W_2}^{m}$ and $\bb{I}_{W_2}^{M}$. There are two different cases.

Case 1. If $\bb{I}_{W_1}^{m}=\emptyset$, then $\bb{I}_{W_2}^{m}=\emptyset$ or $\bb{I}_{W_2}^{M}=\emptyset$. We assume $\bb{I}_{W_2}^{m}=\emptyset$, then $\bb{I}_{W_1}^{M}=\bb{I}_{W_2}^{M}$. For any $W\in\mathscr{H}$ and $x\in\mbR$, we define 
\be
\begin{aligned}
\mD_W^x&:=\{\rho\in\mS\ | \ \tr(W\rho)=x\},\\ 
\mM_W^x&:=\{H\in \mM_d(\bb{C})\ | \ \tr(WH)=x\},
\end{aligned}
\ee
where $\mM_d(\bb{C})$ represents the set of all $d\times d$ complex matrices. We claim that $\mD_{W_1}^M=\mD_{W_2}^M$. If not, we take $\rho\in \mD_{W_1}^M$, but $\rho\notin \mD_{W_2}^M$. If $\tr(W_2\rho)>M$, then $\rho\in\bb{I}_{W_2}^{M}$ and $\rho\notin\bb{I}_{W_1}^{M}$ which leads a contradiction. We assume a moment that $\tr(W_2\rho)<M$. Now we choose $\sigma\in\bb{I}_{W_1}^{M}$, then $\tr(W_1\sigma)>M$ and $\tr(W_2\sigma)>M$. For each $\epsilon\in[0,1]$, we let $\rho_\epsilon:=\epsilon\sigma+(1-\epsilon)\rho$. Then for small enough $\epsilon$, we have $\tr(W_1\rho_\epsilon)>M$ but $\tr(W_2\rho_\epsilon)<M$. We obtain a contradiction again. Therefore, $\mD_{W_1}^M=\mD_{W_2}^M$. According to Lemma $2$ in Ref.~\cite {Li M S2024pra}, we get $\mM_{MI-W_1}^0=\mM_{MI-W_2}^0$. Clearly, $\mM_{MI-W_1}^0$ and $\mM_{MI-W_2}^0$ are respectively orthogonal complement of ${\rm Span}\{MI-W_1\}$ and ${\rm Span}\{MI-W_2\}$ and they have the same dimension $d^2-1$. Thus we have $MI-W_1=\alpha(MI-W_2)$ for some nonzero real number $\alpha$, i.e., $W_1=\alpha W_2+(M-\alpha M)I$. Let $\beta=M-\alpha M$, we have $W_1=\alpha W_2+\beta I$. Note that $W_{1,2}\in\bb{W}_{m,M}$, which reveals $\alpha=1$, i.e, $W_1=W_2$.

Similarly, if $\bb{I}_{W_2}^{M}=\emptyset$, we have $\bb{I}_{W_1}^{M}=\bb{I}_{W_2}^{m}$. In the same way, we can show that there is a nonzero real number $s$ such that $MI-W_1=s(W_2-mI)$, i.e., $W_1=\alpha W_2+\beta I$, where $\alpha=-s$ and $\beta=M+sm$. Since $W_{1,2}\in\bb{W}_{m,M}$, we obtain $s=1$, that is, $W_1+W_2=(m+M)I$.

Case 2. If $\bb{I}_{W_1}^{m}\neq\emptyset$, then  $\bb{I}_{W_1}^{m}\subseteq\bb{I}_{W_1}=\bb{I}_{W_2}^{m}\bigcup\bb{I}_{W_2}^{M}$. Moreover, since $\bb{I}_{W_1}^{m}$ is a convex set, $\bb{I}_{W_1}^{m}\subseteq\bb{I}_{W_2}^{m}$ or $\bb{I}_{W_1}^{m}\subseteq\bb{I}_{W_2}^{M}$. If $\bb{I}_{W_1}^{m}\subseteq\bb{I}_{W_2}^{m}$, we have $\bb{I}_{W_1}^{M}\subseteq\bb{I}_{W_2}^{m}\bigcup\bb{I}_{W_2}^{M}$. Since $\bb{I}_{W_1}^{M}$ is also convex, $\bb{I}_{W_1}^{M}\subseteq\bb{I}_{W_2}^{m}$ or $\bb{I}_{W_1}^{M}\subseteq\bb{I}_{W_2}^{M}$. If $\bb{I}_{W_1}^{M}\subseteq\bb{I}_{W_2}^{m}$, then $\bb{I}_{W_1}^{m}\bigcup\bb{I}_{W_1}^{M}\subseteq\bb{I}_{W_2}^{m}$. This leads to $\bb{I}_{W_1}^{M}=\emptyset$. Therefore, we can further obtain $\bb{I}_{W_1}^{m}=\bb{I}_{W_2}^{m}$ whenever $\bb{I}_{W_1}^{M}\subseteq\bb{I}_{W_2}^{m}$. Using Lemma $2$ in Ref.~\cite {Li M S2024pra} again, we get $W_1-mI=t(W_2-mI)$, that is, $W_1=\alpha W_2+\beta I$ for some $\alpha$ and $\beta=m-\alpha m$. Together with the assumption $W_{1,2}\in\bb{W}_{m,M}$, we obtain $\alpha=1$, i.e, $W_1=W_2$. For the latter case $\bb{I}_{W_1}^{M}\subseteq\bb{I}_{W_2}^{M}$, we have $\bb{I}_{W_1}^{M}=\bb{I}_{W_2}^{M}$ and $\bb{I}_{W_1}^{m}=\bb{I}_{W_2}^{m}$ due to $\bb{I}_{W_1}=\bb{I}_{W_2}$. It also reveals $W_1=W_2$. 

If $\bb{I}_{W_1}^{m}\subseteq\bb{I}_{W_2}^{M}$, we assume first that $\bb{I}_{W_1}^{M}\subseteq\bb{I}_{W_2}^{M}$. Then there a nonzero real number $h$, such that $mI-W_1=h(MI-W_2)$. According to $W_{1,2}\in\bb{W}_{m,M}$, we can conclude that $h=1$, that is, $W_1+W_2=(m+M)I$. When $\bb{I}_{W_1}^{M}\subseteq\bb{I}_{W_2}^{m}$, we can deduce that $\bb{I}_{W_1}^{m}=\bb{I}_{W_2}^{M}$ and $\bb{I}_{W_1}^{M}=\bb{I}_{W_2}^{m}$, which implies $W_1+W_2=(m+M)I$.

Conversely, if $W_1=W_2$, it is obvious that $\bb{I}_{W_1}=\bb{I}_{W_2}$. If $W_1+W_2=(m+M)I$, then $\tr(W_1\rho)<m$ (resp.  $\tr(W_1\rho)>M$) if and only of $\tr(W_2\rho)>M$ (resp. $\tr(W_2\rho)<m$). Hence, $\bb{I}_{W_1}=\bb{I}_{W_2}$.

{\rm(2)} 
Let $W_1':=W_1-aI,W_2':=W_2-aI$. Then
\beax
\bb{I}_{W_1}&=&\{\rho\in\mS\ |\ \tr(W_1\rho)\neq a\}=\{\rho\in\mS\ |\ \tr(W_1'\rho)\neq 0\},\\ 
\bb{I}_{W_2}&=&\{\rho\in\mS\ |\ \tr(W_2\rho)\neq a\}=\{\rho\in\mS\ |\ \tr(W_2'\rho)\neq 0\}.
\eeax 
If $\bb{I}_{W_1}=\bb{I}_{W_2}$, we can obtain $\mD_{W_1'}^0=\mD_{W_2'}^0$. According to Lemma $2$ in Ref.~\cite {Li M S2024pra}, we can show that $W_2'=rW_1'$ for some $r\neq0$, i.e., $W_1-aI=r(W_2-aI)$. Conversely, it is clearly valid.

If $W_1-aI\neq r(W_2-aI)$. We assume with no loss of generality that $\bb{I}_{W_1}\subseteq\bb{I}_{W_2}$. According to the definition of $\mM_W^x$, we have $\mM_{W_2'}^0\subseteq\mM_{W_1'}^0$. Both of them are
linear spaces with the same dimension $d^2-1$. This reveals $\mM_{W_2'}^0=\mM_{W_1'}^0$. By the argument above as in (1), we deduce that $W_2'=rW_1'$, which contradicts the assumption.    
\end{proof}


\section{The finer relation of witnesses} \label{sec7}


Given $W_1,\ W_2\in\bb{W}_{m,M}$, similar to that of the finer relation of the original imaginarity witnesses in Ref.~\cite{Li N 2025 pla}, we say that $W_1$ is finer than $W_2$ if $\bb{I}_{W_2}\subseteq\bb{I}_{W_1}$. 
The finer relation of the original imaginarity witnesses was explored in Ref.~\cite{Li N 2025 pla}. 
In the following, we study such a relation in $\bb{W}_{m,M}$. We begin with the following lemma.

\begin{lemma}\label{L5}
Let $W_1, W_2\in\bb{W}_{m,M}$ with $M\neq m$. 
If $\bb{I}_{W_1}^{M}\subseteq\bb{I}_{W_2}^{m}$, we define
\bea\label{lambda}
\lambda:=\inf_{\rho\in\bb{I}_{W_1}^{M}}\frac{|\tr[(W_2-mI)\rho]|}{|\tr[(MI-W_1)\rho]|}.   
\eea 
Then the following statements are true:

(1) If $\tr(W_1\rho)=M$, then $\tr(W_2\rho)\leq m$.

(2) $\lambda>0$.

(3) If $\tr(W_1\rho)<M$, then $\tr[(W_2-mI)\rho]\leq\lambda\tr[(MI-W_1)\rho]$.
\end{lemma}

\begin{proof}
(1) Suppose $\tr(W_2\rho)>m$, then, for any $\rho_1\in\bb{I}_{W_1}^{M}$ and $a\geq0$, we have     
\beax
\rho_a=(\rho_1+a\rho)/(1+a)\in\bb{I}_{W_1}^{M}.    
\eeax 
On the other hand, there exists a positive number $a_0$ such that $\tr(W_2\rho_a)>m$ holds for all $a\geq a_0$, which is impossible since it leads to $\rho_a\notin\bb{I}_{W_2}^{m}$.

(2) If $\lambda=0$, there exists a sequence $\{\rho_n\}\subseteq\bb{I}_{W_1}^{M}$ such that
\be \label{6}
\epsilon_n=\frac{\tr[(W_2-mI)\rho_n]}{\tr[(MI-W_1)\rho_n]}\to0,\ n\to\infty.  
\ee
Note that there exists a real state $\delta$ such that
both $\tr[(MI-W_1)\delta]>0$ and $\tr[(W_2-mI)\delta]>0$ are nonzero. 
\if false If not, then for any real state $\delta$, either $\tr[(MI-W_1)\delta]=0$ or $\tr[(W_2-mI)\delta]=0$. Thus, by Lemma~\ref{L4}, there exist $\delta_1, \delta_2\in\mR$ so that $\tr[(MI-W_1)\delta_1]=t>0$, $\tr[(MI-W_1)\delta_2]=0$, $\tr[(W_2-mI)\delta_1]=0$ and $\tr[(W_2-mI)\delta_2]=s>0$. Let   $\delta=\frac{s}{t+s}\delta_1+\frac{t}{t+s}\delta_2$. Then $\tr[(MI-W_1)\delta]=\tr[(W_2-mI)\delta_2]=\frac{ts}{t+s}\neq0$,
contradicting to the assumption.
\fi 
We let
\beax \widetilde{\rho}_n:=\frac{1}{1-\frac{\tr[(MI-W_1)\rho_n]}{\tr[(MI-W_1)\delta]}}\left[\rho_n-\frac{\tr[(MI-W_1)\rho_n]}{\tr[(MI-W_1)\delta]}\delta\right] 
\eeax
with $\rho_n$ satisfying Eq.~\eqref{6}. Then $\tr(W_1\widetilde{\rho}_n)=M$, and by item $(1)$, we have 
$\tr(W_2\widetilde{\rho}_n)\leq m$ for each $n$. However,
\begin{widetext}
\beax 
&&\tr[(W_2-mI)\widetilde{\rho}_n]\\
&=&\frac{1}{1-\frac{\tr[(MI-W_1)\rho_n]}{\tr[(MI-W_1)\delta]}}\left[\tr[(W_2-mI)\rho_n]-\frac{\tr[(MI-W_1)\rho_n]}{\tr[(MI-W_1)\delta]}\tr[(W_2-mI)\delta]\right]\\
&=&\frac{1}{1-\frac{\tr[(MI-W_1)\rho_n]}{\tr[(MI-W_1)\delta]}}\left[\epsilon_n-\frac{\tr[(W_2-mI)\delta]}{\tr[(MI-W_1)\delta]}\right]\tr[(MI-W_1)\rho_n],
\eeax  
\end{widetext}
and $\epsilon_n\to0$, which implies that for sufficiently large $n$, $\epsilon_n<\frac{\tr[(W_2-mI)\delta]]}{\tr[(MI-W_1)\delta]}$
and hence $\tr(W_2\widetilde{\rho}_n)>m$, a contradiction.

(3) If $\tr(W_1\rho)<M$, we take $\rho_1\in\bb{I}_{W_1}^{M}$ and let 
\begin{widetext}
\beax
\widetilde{\rho}:=\frac{1}{\tr[(MI-W_1)\rho]-\tr[(MI  -W_1)\rho_1]}\left[\tr[(MI-W_1)\rho]\rho_1\right. 
\left. -\tr[(MI-W_1)\rho_1]\rho\right].
\eeax
\end{widetext}
Then we have
$\tr(W_1\widetilde{\rho})=M$. By item (1), we obtain that $\tr(W_2\widetilde{\rho})\leq m$, and thus 
\begin{widetext}
\beax 
\tr[(MI-W_1)\rho]\tr[(W_2-mI)\rho_1]\leq\tr[(MI-W_1)\rho_1]\tr[(W_2-mI)\rho].
\eeax  
\end{widetext}
It follows that     
\bea
\frac{\tr[(W_2-mI)\rho]}{\tr[(MI-W_1)\rho]}\leq\frac{|\tr[(W_2-mI)\rho_1]|}{|\tr[(MI-W_1)\rho_1]|}.   
\eea 
Taking the infimum with respect to $\rho_1\in\bb{I}_{W_1}^{M}$ on the right side
of the above equation, we get $\tr[(W_2-mI)\rho]\leq\lambda\tr[(MI-W_1)\rho]$.  	
\end{proof}

Similarly, for the other cases of $\bb{I}_{W_1}^{M}\subseteq\bb{I}_{W_2}^{M}$, $\bb{I}_{W_1}^{m}\subseteq\bb{I}_{W_2}^{m}$, and $\bb{I}_{W_1}^{m}\subseteq\bb{I}_{W_2}^{M}$, respectively, we have the following:

(1) For $\bb{I}_{W_1}^{M}\subseteq\bb{I}_{W_2}^{M}$, we define
\bea
\eta:=\inf_{\rho\in\bb{I}_{W_1}^{M}}\frac{|\tr[(MI-W_2)\rho]|}{|\tr[(MI-W_1)\rho]|}.   
\eea
Then, (i) if $\tr(W_1\rho)=M$, then $\tr(W_2\rho)\geq M$; (ii) $\eta>0$; (iii) if $\tr(W_1\rho)<M$, then $\tr[(MI-W_2)\rho]\leq\eta\tr[(MI-W_1)\rho]$.

(2) For $\bb{I}_{W_1}^{m}\subseteq\bb{I}_{W_2}^{m}$, we define
\bea
\xi:=\inf_{\rho\in\bb{I}_{W_1}^{m}}\frac{|\tr[(W_2-mI)\rho]|}{|\tr[(W_1-mI)\rho]|}.   
\eea
Then, (i) if $\tr(W_1\rho)=m$, then $\tr(W_2\rho)\leq m$; (ii) $\xi>0$; (iii) if $\tr(W_1\rho)>m$, then $\tr[(W_2-mI)\rho]\leq\xi\tr[(W_1-mI)\rho]$.

(3) For $\bb{I}_{W_1}^{m}\subseteq\bb{I}_{W_2}^{M}$, we define
\bea
\zeta:=\inf_{\rho\in\bb{I}_{W_1}^{m}}\frac{|\tr[(MI-W_2)\rho]|}{|\tr[(W_1-mI)\rho]|}.   
\eea
Then, (i) if $\tr(W_1\rho)=m$, then $\tr(W_2\rho)\geq M$; (ii) $\zeta>0$; (iii) if $\tr(W_1\rho)>m$, then $\tr[(MI-W_2)\rho]\leq\zeta\tr[(W_1-mI)\rho]$.

\begin{theorem}\label{th4}
Let $W_1, W_2\in\bb{W}_{m,M}$. If $\bb{I}_{W_1}\subseteq\bb{I}_{W_2}$, then there exists real numbers $\alpha$, $\beta$, $\gamma$ and a positive operator $P$, such that $W_1=\alpha W_2+\beta I+\gamma P$, where $\alpha\neq 0$ and $\gamma=1\ \text{or}\ -1$. In particular, if $m=M$ additionally, $W_1=\alpha W_2+\beta I$.
\end{theorem}

\begin{proof}
There are two cases: (1) $m\neq M$ and (2) $m=M$. 

(1) We assume that $m\neq M$. Since $\bb{I}_{W_1}=\bb{I}_{W_1}^{m}\bigcup\bb{I}_{W_1}^{M},\bb{I}_{W_2}=\bb{I}_{W_2}^{m}\bigcup\bb{I}_{W_2}^{M}$,  where $\bb{I}_{W_1}^{m}$ and $\bb{I}_{W_1}^{M}$ are two disjoint convex sets, so are $\bb{I}_{W_2}^{m}$ and $\bb{I}_{W_2}^{M}$.

Case 1. $\bb{I}_{W_1}^{m}=\emptyset$. In such a case, $\bb{I}_{W_1}^{M}\subseteq\bb{I}_{W_2}^{m}$ or $\bb{I}_{W_1}^{M}\subseteq\bb{I}_{W_2}^{M}$. If $\bb{I}_{W_1}^{M}\subseteq\bb{I}_{W_2}^{m}$, by Lemma~\ref{L5},  
\bea
\tr[(W_2-mI)\rho]\leq\lambda\tr[(MI-W_1)\rho]
\eea 
holds for all quantum states $\rho$, where $\lambda$ is defined as in Eq.~\eqref{lambda}.
This implies that $D_1=\lambda(MI-W_1)-(W_2-mI)\geq0$. Let $P_1=\lambda^{-1}D_1$ and $r=\lambda^{-1}$, then $MI-W_1=r(W_2-mI)+P_1$. That is, $W_1=\alpha W_2+\beta I+\gamma P$, where $\alpha=-r$, $\beta=rm+M$ and $\gamma=-1$. 

Similarly, if $\bb{I}_{W_1}^{M}\subseteq\bb{I}_{W_2}^{M}$, we can show that there exists $s>0$ and $P_2\geqslant0$, such that $MI-W_1=s(MI-W_2)+P_2$, i.e., $W_1=\alpha W_2+\beta I+\gamma P$, where $\alpha=s$, $\beta=M-sM$ and $\gamma=-1$.

Case 2. If $\bb{I}_{W_1}^{m}\neq\emptyset$, then  $\bb{I}_{W_1}^{m}\subseteq\bb{I}_{W_1}\subseteq\bb{I}_{W_2}^{m}\bigcup\bb{I}_{W_2}^{M}$. Moreover, since $\bb{I}_{W_1}^{m}$ is a convex set, we can obtain $\bb{I}_{W_1}^{m}\subseteq\bb{I}_{W_2}^{m}$ or $\bb{I}_{W_1}^{m}\subseteq\bb{I}_{W_2}^{M}$. If $\bb{I}_{W_1}^{m}\subseteq\bb{I}_{W_2}^{m}$, we have $\bb{I}_{W_1}^{M}\subseteq\bb{I}_{W_2}^{m}\bigcup\bb{I}_{W_2}^{M}$. Since $\bb{I}_{W_1}^{M}$ is also a convex set,  $\bb{I}_{W_1}^{M}\subseteq\bb{I}_{W_2}^{m}$ or $\bb{I}_{W_1}^{M}\subseteq\bb{I}_{W_2}^{M}$. If $\bb{I}_{W_1}^{M}\subseteq\bb{I}_{W_2}^{m}$, then $\bb{I}_{W_1}^{m}\bigcup\bb{I}_{W_1}^{M}\subseteq\bb{I}_{W_2}^{m}$. Obviously, $\bb{I}_{W_1}^{M}=\emptyset$. Therefore, we can further obtain $\bb{I}_{W_1}^{m}\subseteq\bb{I}_{W_2}^{m}$ whenever $\bb{I}_{W_1}^{M}\subseteq\bb{I}_{W_2}^{m}$, so there exists a $t>0$ and $P_3\geqslant0$, such that $W_1-mI=t(W_2-mI)+P_3$, i,e, $W_1=\alpha W_2+\beta I+\gamma P$ where $\alpha=t$, $\beta =m-tm$ and $\gamma=1$. For the latter case $\bb{I}_{W_1}^{M}\subseteq\bb{I}_{W_2}^{M}$, we have  $\bb{I}_{W_1}^{m}\subseteq\bb{I}_{W_2}^{m}$. Repeat the above discussion, we get the same conclusion.

If $\bb{I}_{W_1}^{m}\subseteq\bb{I}_{W_2}^{M}$, we first assume that $\bb{I}_{W_1}^{M}\subseteq\bb{I}_{W_2}^{M}$. Then there exists $h>0$ and $P_4\geqslant0$ such that $W_1-mI=h(MI-W_2)+P_4$, that is, $W_1=\alpha W_2+\beta I+\gamma P$ where $\alpha=-h$, $\beta=hM+m$, and $\gamma=1$. When $\bb{I}_{W_1}^{M}\subseteq\bb{I}_{W_2}^{m}$, there are positive real numberss $u$, $v$ and  $P_5\geqslant0$ and $P_6\geqslant0$, such that $MI-W_1=u(W_2-mI)+P_5$ and $W_1-mI=v(MI-W_2)+P_6$. That is, we still get $W_1=\alpha W_2+\beta I+\gamma P$ for some $\alpha$, $\beta$ and $\gamma$.

(2) If $m=M$, according to the second part of Theorem~\ref{th3} $\rm(2)$, it is obvious that the statement holds. 
\end{proof}

Theorem~\ref{th4} is different from the corresponding result for both entanglement witnesses~\cite{Wu Y C2006pla,Hou J2010pra,Lewenstein2000pra} and coherence witnesses~\cite{Li M S2024pra,Wang2021qip} in which there are both sufficient and necessary conditions for the finer relation of two witnesses. We now give a counterexample to show that, the converse of this theorem is not true.

\begin{example}
Let 
\beax
W_1=\begin{pmatrix}
	1&\rm i&\rm i\\
	-\rm i&0&0\\
	-\rm i&0&1\\
\end{pmatrix},~
W_2=\begin{pmatrix}
	1&\rm i&0\\
	-\rm i&0&0\\
	0&0&0\\
\end{pmatrix},
\eeax
\beax
P=\begin{pmatrix}
	2&0&\rm i\\
	0&2&0\\
	-\rm i&0&3\\
\end{pmatrix}. 
\eeax
Then $W_1,\ W_2\in\bb{W}_{0,1}$ and $W_1=W_2-2I+P$. Taking
\beax
\rho=\begin{pmatrix}
	{1}/{2}&{\rm i}/{6}&{\rm i}/{6}\\
	-{\rm i}/{6}&{1}/{6}&{\rm i}/{6}\\
	-{\rm i}/{6}&-{\rm i}/{6}&{1}/{3}\\
\end{pmatrix},
\eeax
we have $\tr(W_1\rho)=\frac{3}{2}\notin[0,1]$ and $\tr(W_2\rho)=\frac{5}{6}\in[0,1]$. Hence, $\rho\in\bb{I}_{W_1}$, but $\rho\notin\bb{I}_{W_2}$, that is, $\bb{I}_{W_1}\nsubseteq\bb{I}_{W_2}$.
\end{example}

\begin{example}
Let 
\beax
W_1=\begin{pmatrix}
	1&0&\rm i\\
	0&1&\rm i\\
	-\rm i&-\rm i&1\\
\end{pmatrix},~
W_2=\begin{pmatrix}
	1&\rm i&0\\
	-\rm i&1&0\\
	0&0&1\\
\end{pmatrix},
\eeax
\beax			
P=\begin{pmatrix}
	2&-\rm i&\rm i\\
	\rm i&2&\rm i\\
	-\rm i&-\rm i&2\\
\end{pmatrix}.			
\eeax 
Then $W_1,\ W_2\in\bb{W}_{1,1}$, $W_1=W_2-2I+P$. We consider
\beax
\rho=\begin{pmatrix}
	{1}/{2}&{\rm i}/{6}&0\\
	-{\rm i}/{6}&{1}/{3}&0\\
	0&0&{1}/{6}\\
\end{pmatrix},
\eeax
which leads to  
\beax 
\tr(W_1\rho)=1\ \text{and}\ \tr(W_2\rho)=\frac{4}{3}\neq1.
\eeax 
i.e., $\bb{I}_{W_1}\nsubseteq\bb{I}_{W_2}$.	
\end{example}

As in Ref.~\cite {Li N 2025 pla}, we call $W$ an optimal imaginarity witness in $\bb{W}_{m,M}$ if there is no imaginarity witness that can be finer than $W$. Next, we discuss the conditions under which the witness in $\bb{W}_{m,M}$ is optimal. By Theorem~\ref{th4}, we directly derive the following.

\begin{theorem}
Let $W\in\bb{W}_{m,M}$. If for any positive operator $P$ and real numbers $\alpha$, $\beta$, $\gamma$, where $\alpha\neq0$, such that $\alpha W+\beta I+\gamma P\notin\bb{W}_{m,M}$, then $W$ is optimal. 	
\end{theorem}


\section{Witnessed imaginarity} \label{sec8}


In the resource theory of entanglement and coherence, witnesses can be used for quantifying the ``amount'' of the resource contained in the resource state~\cite{Ren H2017alp,Brandao2005pra}, which were called witnessed entanglement~\cite{Brandao2005pra} and witnessed coherence~\cite{Ren H2017alp}, respectively. They always coincide with some other known measures. For example, both the negativity and concurrence can be represented by some special families of entanglement witnesses~\cite{Brandao2005pra,Verstraete2002thesis}, and the $l_1$-norm of coherence is just the witnessed coherence~\cite{Ren H2017alp}. We discuss below whether the imaginarity witnesses can educe an imaginarity measure.

We define the imaginarity witness by
\be 
\mI_W(\rho)=\max_{\|W\|\leq1}I_W(\rho),~\forall\rho\in\mS,
\ee
where 
\be 
I_W(\rho)=\max\{0,m-\tr(W\rho)\}+\max\{0,\tr(W\rho)-M\},
\ee 
${m}$ and ${M}$ denote the minimum and maximum eigenvalues value of $\rRe(W)$ respectively, 
and where the maximum runs over all imaginarity witnesses $W$'s with $\|W\|\leq1$.

\begin{lemma}\label{L6}
For any nonreal Hermitian $X$ under the reference basis and any $H\in\mathscr{H}$,
\bea\max_{\|W\|\leq1}\tr(WX)=\max_{\|H\|\leq1}\tr(HX)=\|X\|_{\tr}.
\eea
\end{lemma}

\begin{proof}
For any nonreal Hermitian $X$ and any $H\in\mathscr{H}$, obviously, $\tr(WX)\leq\|X\|_{\tr}$ whenever $\|W\|\leq1$ and $\tr(HX)\leq\|X\|_{\tr}$ whenever $\|H\|\leq1$. We let the spectral decomposition of $X$ be $X=\sum_{i}\lambda_{i}|v_i\ra\la v_i|$. Taking 
\beax W_0=\sum_{i}{\rm sgn}(\lambda_{i})|v_{i}\ra\la v_{i}|,
\eeax  
$\tr(W_0X)=\sum_{i}|\lambda_{i}|=\|X\|_{\tr}$; taking 
\beax 
H_0=\sum_{i}{\rm sgn}(\lambda_{i})|v_{i}\ra\la v_{i}|,
\eeax $\tr(H_0X)=\|X\|_{\tr}$. The proof is completed.	
\end{proof}	   

\begin{theorem}
For any $\rho\in\mS$, 
\be 
\mI_W(\rho)=\mI_{tr}(\rho)=\mI_{R}(\rho).
\ee
\end{theorem}

\begin{proof}
If $W\in\bb{W}_{m,M}$ for some $m_W$ and $M_W$, we know $M_W=\max_{\delta\in\mR}\tr(W\delta)$ and $m_W=\min_{\delta\in\mR}\tr(W\delta)$.
If $\tr(W\rho)>M_W$, then 
\bea 
I_W(\rho)=\tr(W\rho)-M_W=\min_{\delta\in\mR}\tr[W(\rho-\delta)].
\eea  
Since the sets of all Hermitian operators $H$'s with $\|H\|\leq1$ and $\mR$ are convex and compact, by von Neumann's minimax theorem~\cite{{Nikaid},{Sion}}, we get 
\bea \max_{\|H\|\leq1}\min_{\delta\in\mR}\tr(H\delta)=\min_{\delta\in\mR}\max_{\|H\| \leq1}\tr(H\delta).
\eea  
By Lemma~\ref{L6}, $\max_{\|W\| \leq1}\tr[W(\rho-\delta)]=\max_{\|H\| \leq1}\tr[H(\rho-\delta)]=\|\rho-\delta\|_{\tr}$, and together with the fact that any nonreal Hermitan $X$ is an imaginarity witness, we have
\beax 
\mI_W(\rho)&=&\max_{\|W\|\leq1}I_W(\rho)\\
&=&\max_{\|W\|\leq1}\min_{\delta\in\mR}\tr[W(\rho-\delta)]\\
&=&\max_{\|H\|\leq1}\min_{\delta\in\mR}\tr[H(\rho-\delta)]\\
&=&\min_{\delta\in\mR}\max_{\|H\|\leq1}\tr[H(\rho-\delta)]\\
&=&\min_{\delta\in\mR}\|\rho-\delta\|_{\tr}=\mI_{tr}(\rho).
\eeax
Similarly, if $\tr(W\rho)<m_W$, we also can obtain $\mI_W(\rho)=\mI_{tr}(\rho)$. Together with the fact $\mI_{tr}(\rho)=\mI_{R}(\rho)$ in Ref.~\cite{WuKD2021pra}, we finish the proof. 
\end{proof}


\section{Conclusion and discussion}\label{sec9}


In this work, we have proposed a new type of imaginarity witnesses, in which the expect value of all the real states are lie in the interval that determined by its minimum eigenvalue and maximum eigenvalue of the real part of the given witness. Any state with the expect value is not included in this interval reveals the imaginarity contained in the state. In such a sense, any nonreal Hermitian operator under the given reference basis is a imaginarity witness. On the other hand, only a finite number of such imaginarity witnesses can detect all the imaginarity states. This type of imaginarity witnesses outperform the previous ones and cover the stringent imaginarity witnesses as a special case.

As the witness theory for other resource such as the entanglement and coherence, we also explored different relations between or among these imaginarity witnesses, which including, when different witnesses can detect common imaginarity states or the same imaginarity states and the finer relation of witnesses. Consequently, we found that, the conclusions for these imaginarity witnesses are a little different from that of entanglement witnesses and the coherence witnesses.

Finally, we investigated the concept of witnessed imaginarity, establishing it as the natural counterpart to both witnessed entanglement and witnessed coherence. Notably, our analysis demonstrates that witnessed imaginarity is equivalent to both the trace norm of imaginarity and the robustness of imaginarity. Building upon existing research regarding imaginarity witnesses in the literature, these findings collectively constitute a comprehensive theoretical framework for understanding imaginarity witnesses.

\begin{acknowledgements}
This work is supported by the National Natural Science Foundation of China under Grant Nos.~12471434 and 11971277, the Program for Young Talents of Science and Technology in Universities of Inner Mongolia Autonomous Region under Grant No. NJYT25010, and the High-Level Talent Research Start-up Fund of Inner Mongolia University under Grant No. 10000-A260015/501.
\end{acknowledgements}
	

	
	

\end{document}